\documentclass[draftcls, onecolumn,12pt]{IEEEtran}

\usepackage{amsmath}
\usepackage{cases}
\usepackage{amssymb}
\usepackage{graphicx}
\usepackage{extarrows}
\usepackage{bm}

\usepackage[english]{babel}
\usepackage{algorithm} 
\usepackage{algorithmic} 
\usepackage{multirow}

\newtheorem{theorem}{Theorem}
\newtheorem{lemma}{Lemma}

\newtheorem{definition}{Definition}

\begin{document}
\baselineskip=22pt
\title{High Speed Railway Wireless Communications: Efficiency v.s. Fairness \footnote{ \hrule\quad\\ This paper has been published by \textit{IEEE Trans. on Veh. Technol.} in Feb. 2014.} }

\author{Yunquan~Dong,
        Pingyi~Fan,~\IEEEmembership{Senior Member,~IEEE} ~and~Khaled~Ben~Letaief$^{\dag}$,~\IEEEmembership{Fellow,~IEEE},\\
        \{dongyq08@mails,~fpy@mail\}.tsinghua.edu.cn, ~$^{\dag}$eekhaled@ece.ust.hk }

\maketitle

\begin{abstract}

High speed railways (HSRs) have been deployed widely all over the world in recent years. Different from traditional cellular communication, its high mobility makes it essential to implement power allocation along the time. In the HSR case, the transmission rate depends greatly on the distance between the base station (BS) and the train. As a result, the train receives a time varying data rate service when passing by a BS. It is clear that the most efficient power allocation will spend all the power when the train is nearest from the BS, which will cause great unfairness along the time. On the other hand, the channel inversion allocation achieves the best fairness in terms of constant rate transmission. However, its power efficiency is much lower. Therefore, the power efficiency and the fairness along time are two incompatible objects. For the HSR cellular system considered in this paper, a trade-off between the two is achieved by proposing a temporal proportional fair power allocation scheme. Besides, near optimal closed form solution and one algorithm finding the $\epsilon$-optimal allocation are presented.

\end{abstract}


\let\thefootnote\relax\footnotetext{Copyright (c) 2013 IEEE. Personal use of this material is permitted. However, permission to use this material for any other purposes must be obtained from the IEEE by sending a request to pubs-permissions@ieee.org. }

\section{Introduction}

In recent years, high speed trains with an operation speed of more than 300 km/h are being deployed rapidly all over the world. Particularly, bullet trains and high speed trains have been widely used in China. Besides the high mobility of the train, the growing demands on high speed data services make it crucial to investigate the specific channel and develop efficient transmission schemes for high speed railway (HSR) communications. Among the research on HSR, the authors \cite{Zhao-2013} proposed a three-dimensional model of line-of-sight multi-input-mulit-output (LOS MIMO)  channel for HSR viaduct environment and analyzed its MIMO channel capacity. New channel estimation technique was proposed for long-term evolution (LTE) systems in the HSR environment in \cite{Qiu-2012}.
Moreover, formalization of a complete dynamic model that represents the dynamic coupling of electrical and mechanical phenomena dedicated to the HSR systems was presented in \cite{Colombaioni-2008}. New system architectures based on radio over fiber was also proposed in \cite{Zhou-2011}, which increases throughput and decreases handovers. However, there are some inevitable increases in complexity.

Particularly, there come out two basic problems due to the high speed of trains in HSR cellular systems. Firstly, since base stations (BSs) of finite coverage are positioned along the railway, the time that a certain BS can serve the train is limited. As a result, besides the data rate,
we must take in to account of the length of the serving period of the a in HSR systems. This makes the concept of channel service \cite{DWF-Deterministic-TWC} a good characterization in HSR communications. Secondly, the transmission rate is highly determined by the distance between the BS and the train, which varies quickly as the train moves. However, the train moves at a constant velocity in most cases. Thus, the position as well as the transmission rate at the next moment can be predicted. This unique feature of HSR communication makes it necessary and feasible to implement power allocation along the time at the BS.


Define channel service as the total amount of service provided by the channel in a period of $t$ \cite{DWF-Deterministic-TWC}, where $t$ is the duration in which the train is moving in the coverage of the BS.  Then $S(t)$ is equal to the sum of the service provided at every epochs and can be expressed as the integral of the instantaneous capacity $C(\tau)$, i.e., $S(t)=\int_0^t C(\tau) \mathrm{d}\tau$. Particularly, it has been applied to the base station arrangement in HSR cellular systems in \cite{Chuang-bs_arrange}, where the authors considered the most efficient range of BS's service in the BS deployment.

On one hand, for some given average power, larger $S(t)$ means higher energy efficiency and better user experience. On the other hand, the signal pathloss is increasing with the distance between the BS and the train. If the protocol is designed simply to maximize $S(t)$, the power consumption will mainly concentrate on the period when the train is nearest from the BS, which will cause serious unfairness along the time.

A possible trade-off between channel service and fairness along time is the proportional fairness (PF), which stems from game theory and gets wide applications in the Qualcomm’s HDR (high data rate) systems \cite{Kelly-1998,Kim-2005}. In this paper, a proportional fair power allocation along time is proposed for HSR communication systems. A near optimal power allocation in closed expression is given. An algorithm that will find the $\epsilon$-optimal solution within any given expected error will also be presented.

The rest of this paper is organized as follows. The system model and some simple power allocation schemes are presented in Section \ref{sec:2}. The definitions of proportional fairness and channel service are also reviewed in this section. Next, the proportional fairness along time is proposed for HSR systems in Section \ref{sec:3}, where a near optimal PF power allocation as well as an algorithm to find the $\epsilon$-optimal PF power allocation is given. The obtained results will be presented via numerical results in Section \ref{sec:4}.  Finally, conclusions on the work are given in section \ref{sec:5}.

\section{System Model}\label{sec:2}
Consider a single antenna communication system for the HSR as shown in Fig \ref{fig:tx_m}, where an elevated antenna on the top of the train serves as the access point for users in the train. Base stations are positioned along the railway, spacing each other by $2R_0$, where $R_0$ is the cellular radius. Suppose the train moves at velocity $v$. Then it takes the train $T=\frac{R_0}{v}$ to pass through the cell from the cell center. For the coordinate system in Fig. \ref{fig:tx_m}, let $d(\tau)$ be the distance between BS and the train at time $\tau$ for $0\leq\tau\leq T$. Assuming that the minimum distance between BSs and the railway is $d_0$, we have $d(\tau)=\sqrt{d_0^2+v^2\tau^2}$. The wireless channel between BS and the train is assumed to be an additive white Gaussian noise channel (AWGN) with LOS pathloss for the following reasons. Firstly, viaducts account for the vast majority of Chinese HSR (more specifically, 86.5$\%$, elevated in Beijing-Shanghai HSR) \cite{Zhao-2013}.
Due to lack of scatters, the received signal at the train is not rich in independent signal paths, among which the LOS path persists the most power. Therefore, the AWGN or sometimes the Rician channel model is suitable. Secondly, it has been proved in \cite{DWF-Deterministic-TWC} that for independent identical distributed fading channels, the service provided by the channel is a deterministic time-linear function, just like the AWGN channel. It is also assumed that the frequency offset estimation and correction are perfect \cite{YF-DFO}, which makes it easy for us to focus on the power allocation problem. At the same time, the power allocation problem in fading HSR systems is also very important and will be considered in future work. Let $\bar{P}$ be the average transmit power of BS, $W$ be the limited signal bandwidth, $N_0$ be the noise power spectral density and $\alpha$ be the pathloss exponent. Suppose the instantaneous transmit power is  $P(\tau)$ according to some power allocation scheme. Besides, every power allocation policy  $P(\tau)$ satisfies $\frac{1}{T}\int_0^T P(\tau)\mathrm{d}\tau=\bar{P}$. Denote $N(\tau)=WN_0d^\alpha(\tau)$, the instantaneous capacity of the channel between BS and the train at time $\tau$ is
\begin{equation}\label{eq:C_tau}\nonumber
    C^{(P)}(\tau)=W\log\left(1+\frac{P(\tau)}{N(\tau)}\right).
\end{equation}
\begin{figure}[!t]
\centering
\includegraphics[width=2.6in]{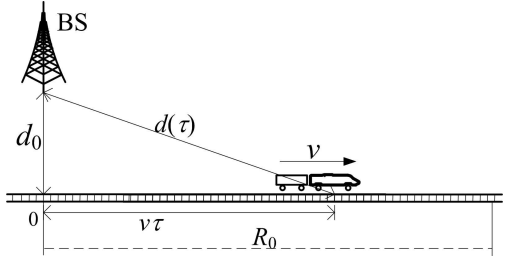}
\caption{The HSR communication system} \label{fig:tx_m}
\end{figure}

\subsection{Constant and channel inversion power allocation}
The most straightforward scheme is the constant power allocation, in which BS maintains a constant transmit power all the time, i.e., $P(\tau)=\bar{P}$. Thus,
\begin{equation}\label{eq:Ctau_cons}
    C^{(P_{con})}(\tau)=W\log\left(1+\frac{\bar{P}}{N(\tau)}\right).
\end{equation}

Since this scheme ignores the variation of instantaneousness channel gain, transmission rate becomes very low when the train is far from the BS, resulting in great unfairness.

Besides, the channel inversion power allocation tries to maintain a constant transmission rate at the transmitter, at all times. Therefore, the ratio between $P(\tau)$ and $N(\tau)$ is a constant. Suppose $P(\tau)=k_0N(\tau)$, by solving $\int_0^T P(\tau) \mathrm{d}\tau = \bar{P}T$
we have $k_0=\frac{\bar{P}T}{\int_0^T N(\tau) \mathrm{d}\tau} = \frac{\bar{P}}{\bar{N}}$ and 

\begin{equation}\label{eq:Ctau_inv}
    C_{inv}(\tau)=W\log(1+k_0).
\end{equation}

In this case, the power efficiency is very low because too much of the power is used to compensate those very bad channel states. However, the best fairness in terms of stable transmission rate is achieved, at the cost of power efficiency.

\subsection{Water filling power allocation}
On the other hand, the maximum channel service can be achieved by water filling scheme. Specifically, it is the solution of the following optimization problem.
\begin{equation}\label{Prob:11}
    \begin{aligned}
        &\underset{P(\tau)}{\max} & & \int_0^T C^{(P)}(\tau) \mathrm{d}\tau \\
        & s.t.                    & & \frac{1}{T}\int_0^T P(\tau) \mathrm{d}\tau=\bar{P},~ P(\tau)\geq0\\
    \end{aligned}
\end{equation}

Let's consider the Lagrangian
    \begin{equation}\label{dr:LagF_0}\nonumber
    \begin{split}
        F_w=&\int_0^T  C^{(P)}(\tau) \mathrm{d}\tau -\lambda\left[\int_0^T P(\tau) \mathrm{d}\tau-\bar{P}T\right]\\
        =&\int_0^T \left[ C^{(P)}(\tau) -\lambda P(\tau) - \lambda\bar{P}\right]\mathrm{d}\tau \triangleq \int_0^T L_w(\tau) \mathrm{d}\tau
    \end{split}
    \end{equation}
where $L_w(\tau)=C^{(P)}(\tau) -\lambda P(\tau) - \lambda\bar{P}$.
According to Euler's Formula, taking the derivative of $L_w(\tau)$ is equivalent to taking the derivative of its integrand, namely
    \begin{equation}\nonumber
            \frac{\partial L_w(\tau)}{\partial P(\tau)}=\frac{W}{P(\tau)+N(\tau)}-\lambda.
    \end{equation}

Set $\frac{\partial L_w}{\partial P(\tau)}=0$, we get
\begin{equation}\label{eq:Ctau_wf}
    P(\tau)=\frac{W}{\lambda}-N(\tau).
\end{equation}

Since $C(\tau)$ decreases monotonically as the train leaves BS, it can be seen that the optimal solution lies on edge of the feasible allocation region, i.e., there is sometime when the allocated power is zero. Suppose $P(\tau)=0$ when $\tau>t_1$ for some $t_1<T$, one has $\int_0^{t_1} P(\tau) \mathrm{d}\tau=\bar{P}T$ and $P(t_1)=0$.
Then $P(\tau)$ can be obtained after solving $t_1$ and $\lambda$.

\subsection{Proportional fairness and channel service}
The proportional fairness power allocation was firstly proposed to balance the average throughput and user fairness in the HDR systems. From the context of the game theory, the definition of the fairness among users is given as follows.
\begin{definition} \cite{Kelly-1998} {Proportional Fairness}\\
     A scheduling $\mathcal{P}$ is said to be proportionally fair if and only if, for any feasible scheduling $\mathcal{S}$, we have:
    \begin{equation}\label{df:pf_org}
        \sum_{i\in U}\frac{R_i^{(\mathcal{S})}-R_i^{(\mathcal{P})}}{R_i^{(\mathcal{P})}}\leq0,
    \end{equation}
    where $U$ is the user set, and $R_i^{(\mathcal{S})}$, is the average rate of user
    $i$ by scheduler $\mathcal{S}$.
\end{definition}

Although clear in physical meaning, this definition is difficult to use when solving the proportional fair scheduling policy. Then the following lemma give us a more mathematically trackable form of proportional fairness.
\begin{lemma}\cite{Kim-2005}
    A scheduling is proportional fair if and only if it is the solution to the following optimization problem

    \begin{equation}\label{prob:1}
        \begin{aligned}
            & \underset{\mathcal{P}}{\max} & & \sum_{i\in U} \ln R_i^{(\mathcal{P})} \\
            & s.t.                         & & \mathcal{P} \text{ is feasible.} \\
        \end{aligned}
    \end{equation}

\end{lemma}

The \textit{channel service} was proposed in \cite{DWF-Deterministic-TWC} to characterize the transmission process over fading channels from a viewpoint of cross-layer designing. Seen from the receiver, it is equal to the amount of data that the receiver can obtain through the channel during a period of $t$, assuming the transmitter has plenty of data at each epoch.
\begin{definition} \cite{DWF-Deterministic-TWC} {Channel Service}\\
    The channel service process $S(t)$ is defined as the amount of service provided by the channel during a period of $t$. For {i.i.d.} (fading) channels, it can be expressed as the integral of the instantaneous channel capacity over the period,
    \begin{equation}\nonumber
        S(t)=\int_0^t C(\tau) \mathrm{d}\tau.
    \end{equation}
\end{definition}

It should be noted that the equation holds in sense of mean square for fading channels.

Since channel service considers both data rate and service time, it is specially applicable to HSR systems. Firstly, for high speed trains, the time that it can be served by a certain BS is limited and highly related to its velocity. Secondly, the date rate that it is served by the channel is closely related to the distance to the BS. Therefore, we have to take two parameters, i.e., rate and time, into consideration at the same time. To this end, channel service is a reasonable choice.

\section{Proportional Fair Power Allocation for HSR}\label{sec:3}
In this paper, proportional fairness will be applied to HSR systems. The difference lies in that we consider the fairness for a fixed user in the time domain, other than the fairness among different users.

\subsection{Proportional fairness along time}
The transmission rate between BS and the train dependents greatly on the distance between the two. However, users in the train are expecting a stable data transmission service as the train moves, i.e., the fairness of service along time. On the other hand, the transmission rate still is an important parameter of concern. Therefore, the contradiction between channel service and fairness along time exists naturally. By applying  proportional fairness to HSR systems,  the proportional fairness along time is defined as follows.
\begin{definition}
    A power allocation $P(\tau)$ is said to be proportionally fair along time if and only if, for any feasible power allocation $S(\tau)$, we have:
    \begin{equation}\label{df:pf_hsr}\nonumber
        \int_0^T\frac{C^{(S)}(\tau)-C^{(P)}(\tau)}{C^{(P)}(\tau)} \mathrm{d}\tau\leq0,
    \end{equation}
    where $C^{(S)}(\tau)$, is the transmission rate at the time $\tau$ by power allocation $S(\tau)$.
\end{definition}

Next, the proportional fair power allocation can be found according to the following Lemma.

\begin{lemma}\label{lem:2_pfhsr}
    A power allocation is proportional fair along time if and only if it is the solution to the following optimization problem
    \begin{equation}\label{Prob:2}
        \begin{aligned}
            &\underset{P(\tau)}{\max} & & \int_0^T \ln C^{(P)}(\tau) \mathrm{d}\tau \\
            & s.t.                    & & \frac{1}{T}\int_0^T P(\tau) \mathrm{d}\tau=\bar{P} \\
        \end{aligned}
    \end{equation}
\end{lemma}

\begin{proof}See Appendix.
\end{proof}

In this formulation, an attractive trade-off is achieved between \textit{channel service} and \textit{fairness along time.} For the HSR system considered in this paper, the proportional fair power allocation is given in the following subsection.

\subsection{The near optimal solution}
Actually, transcendental functions can not be avoided in solving problem (\ref{Prob:2}). Thus, it is quite difficult to obtain the explicit solution, for which an approximation is given as follows.
\begin{theorem}\label{th:solut_app}
    An approximation solution to the proportional fair power allocation for the HSR system is given by
    \begin{equation}\label{rt:th_1}
        P(\tau)=\frac{ [\bar{P}+N(T)]\ln\left(1+\frac{\bar{P}}{N(T)}\right) }{ \textsf{W}\left[\frac{(\bar{P}+N(T))}{N(\tau)}\ln\left(1+\frac{\bar{P}}{N(T)}\right)\right] },
    \end{equation}
    where $N(\tau)=WN_0d^\alpha(\tau)$, $\textsf{W}(z)$ is the solution to $\textsf{W}(z)e^{\textsf{W}(z)}=z$, i.e., the LambertW function \cite{Lambert_w}.
\end{theorem}
\begin{proof}
    Using the standard optimization technique, the corresponding Lagrangian functional is obtained as follows
    \begin{equation}\label{dr:LagF}\nonumber
        \begin{split}
        F=&\int_0^T \ln C^P(\tau) \mathrm{d}\tau -\lambda\left[\int_0^T P(\tau) \mathrm{d}\tau-\bar{P}T\right]\\
        =&\int_0^T \left[ \ln C^P(\tau) -\lambda P(\tau) + \lambda \bar{P} \right] \mathrm{d}\tau
        \end{split}
    \end{equation}
    where $\lambda$ is the Lagrangian multiplier for the constraint.

    Solving this optimization problem needs to differentiate $F$ with respect to $P(\tau)$, which is equivalent to take the derivative of the integrand $L(\tau)$ with respect to $P(\tau)$, according to Euler's Formula. Specifically, we have
    \begin{equation}\nonumber
        \begin{split}
            \frac{\partial L(\tau)}{\partial P(\tau)}=\frac{1}{\ln\left( 1+\frac{P(\tau)}{N(\tau)} \right)} \frac{1}{P(\tau)+N(\tau)}-\lambda.
        \end{split}
    \end{equation}

    Let the derivative equal to zero, $P(\tau)$ can be obtained by the following steps of calculation
    \begin{equation}\label{dr:Ptau}\nonumber
        \begin{split}
            &\left[\ln\left(P(\tau)+N(\tau)\right)-\ln N(\tau)\right]\left[P(\tau)+N(\tau)\right]=\frac{1}{\lambda}\\
            &-\ln N(\tau)= \frac{1}{\lambda\left[\ln\left(P(\tau)+N(\tau)\right)\right]}-\ln\left(P(\tau)+N(\tau)\right)\\
            &\frac{1}{\lambda N(\tau)}=\exp\left( \frac{1}{\lambda\left(P(\tau)+N(\tau)\right)} \right)\frac{1}{\lambda\left(P(\tau)+N(\tau)\right)}\\
            &\frac{1}{\lambda\left(P(\tau)+N(\tau)\right)}=\textsf{W}\left[ \frac{1}{\lambda N(\tau)} \right]
        \end{split}
    \end{equation}
    and finally we have
    \begin{equation}\label{rt:Ptau}
        P(\tau)=\frac{1}{\lambda \textsf{W}\left[ \frac{1}{\lambda N(\tau)} \right]}-N(\tau).
    \end{equation}

    In order to solve the undetermined multiplier $\lambda$, we can use the constraint in (\ref{Prob:2}), i.e., $\int_0^T P(\tau)=\bar{P}T$. Then we have
    \begin{equation}\label{dr:P2_const}
        \int_0^T \frac{1}{\lambda \textsf{W}\left[ \frac{1}{\lambda N(\tau)} \right]} \mathrm{d}\tau=\bar{P}T+\int_0^T N(\tau)\mathrm{d}\tau.
    \end{equation}

    In fact, both the power allocation function $P(\tau)$ and the multiplier $\lambda$ are functions of $T$. However, it is reasonable to assume that they don't change much as $T$ changes very slightly to $T+\Delta T$. Thus, we can take the diversities on the both sides of (\ref{dr:P2_const}) with respect to $T$. Then we have
    \begin{equation}\label{eq:lambda_relex}
        \lambda \textsf{W}\left[ \frac{1}{\lambda N(T)} \right]=\frac{1}{\bar{P}+ N(T)}.
    \end{equation}
    Nevertheless, the assumption here still causes some small deviation from the optimal solution. In this sense, we say the result in Theorem \ref{th:solut_app} is an near optimal solution.

    According to the definition of the LambertW function that $\textsf{W}(z)e^{\textsf{W}(z)}=z$, we have
    \begin{equation}\nonumber
        \begin{split}
            \frac{1}{\lambda N(T)}=&\textsf{W}\left[ \frac{1}{\lambda N(T)} \right] e^{\textsf{W}\left[ \frac{1}{\lambda N(T)} \right]}\\
            =&\frac{1}{\lambda[\bar{P}+ N(T)]} e^{\frac{1}{\lambda[\bar{P}+ N(T)]}},
        \end{split}
    \end{equation}
    from which one gets
    \begin{equation}\label{df:lambda_star}
        \frac{1}{\lambda}=[\bar{P}+N(T)]\ln\left(1+\frac{\bar{P}}{N(T)}\right)\triangleq \frac{1}{\lambda_{apx}}
    \end{equation}
    and finally the closed form of $P(\tau)$ in (\ref{rt:th_1}). This completes the proof.

\end{proof}

\subsection{An algorithm finding the $\epsilon$-optimal solution}
It is seen that the result in Theorem \ref{th:solut_app} is an approximation to the optimal power allocation due to the relaxation in solving $\lambda$ in (\ref{eq:lambda_relex}). In this part, an algorithm is proposed to find more accurate $\lambda$ as well as the power allocation function $P(\tau)$, which can approach the optimal solution within any error range.  

Starting from $\lambda_{apx}$ (see Eqn. \ref{df:lambda_star}), we will find the $\epsilon$-optimal $\lambda_{opt}$ by comparing and stepping. It is seen that the optimal power allocation follows (\ref{rt:Ptau}) and satisfies the constraint in (\ref{Prob:2}).  Denote the total power specified by $P(\tau)$ and $\lambda$ as $P_\lambda=\int_0^T P(\tau) \mathrm{d}\tau|_\lambda$, define the difference ratio $r_{\Delta P}=\frac{P_\lambda}{\bar{P}T}-1$ and sign variable $s_{\Delta P}=\mathrm{sgin}(r_{\Delta P})$, which gives the sign of the total power comparison between the current scheme and the optimal one. Fix the step size of $\lambda$, i.e., $\Delta\lambda$, and the maximum allowable $r_{\Delta P}$ as $\varepsilon_{\Delta P}$, the algorithm is given by Algorithm 1.
\begin{algorithm}[h]
\caption{Searching the $\epsilon$-optimal $\lambda$}
\begin{algorithmic}[1]\label{Alg:1}
\REQUIRE ~~\\
    Set $\lambda=\lambda_{apx}$, $\Delta\lambda$, $\varepsilon_{\Delta P}$,\\
    ~~~~~$P_{\lambda_{apx}}=\int_0^T P(\tau) \mathrm{d}\tau|_{\lambda_{apx}}$,\\
    ~~~~~$r_{\Delta P}=\frac{P_{\lambda_{apx}}}{\bar{P}T}-1$;
\ENSURE ~~\\

\WHILE {$|r_{\Delta P}|>\varepsilon_{\Delta P}$}
\STATE $s_{\Delta P}=\mathrm{sgin}(r_{\Delta P})$;
\STATE $\lambda=\frac{1}{\frac{1}{\lambda}-s_{\Delta P}\cdot\Delta\lambda}$;
\STATE $P_\lambda=\int_0^T P(\tau) \mathrm{d}\tau|_{\lambda}$;
\STATE $r_{\Delta P}=\frac{P_{\lambda}}{\bar{P}T}-1$;
      \IF {~$\mathrm{sign}(r_{\Delta P})\cdot s_{\Delta P}>0$}
        \STATE $\Delta\lambda=2\Delta\lambda$
      \ELSE
        \STATE $\Delta\lambda=\Delta\lambda/7$
      \ENDIF
\ENDWHILE \label{code:recentEnd}

\STATE  $\lambda_{opt}$, the $\epsilon$-optimal power allocation:~$P_{\lambda_{opt}}(\tau)$
\end{algorithmic}
\end{algorithm}

As shown, $\lambda$ will be increased in step 3 if $s_{\Delta P}=1$ in step 2 and will be decreased if $s_{\Delta P}=-1$ in step 2. The reason is that $P(\tau)$ is monotonically decreasing with $\lambda$, which can be checked by its derivative,
\begin{equation}\nonumber
    \begin{split}
        &\frac{\partial P(\tau)}{\partial \lambda}=\frac{\partial}{\partial \lambda}\left[\frac{1}{\lambda \textsf{W}\left[ \frac{1}{\lambda N(\tau)} \right]}-N(\tau)\right]\\
        =&\frac{-1}{\lambda^2 \textsf{W}\left[ \frac{1}{\lambda N(\tau)} \right]}+ \frac{1}{\lambda^2\textsf{W}\left[ \frac{1}{\lambda N(\tau)} \right]\left(1+\textsf{W}\left[ \frac{1}{\lambda N(\tau)} \right]\right)}\\
        =&\frac{-1}{\lambda^2 \left(1+\textsf{W}\left[ \frac{1}{\lambda N(\tau)} \right]\right)}<0.
    \end{split}
\end{equation}

From (\ref{df:lambda_star}), it is seen that the value of $\frac{1}{\lambda}$ is usually larger than itself. Thus, we are actually changing the value of $\frac{1}{\lambda}$ in the third step of the loop, which makes the program more smooth. The step size of $\lambda$ will also be changed adaptively (step 6 to 10). Specifically, if $r_{\Delta P}>0$ in two adjacent loops, we will increase the step size so that the total power can approach $\bar{P}T$ more quickly. Otherwise, if $r_{\Delta P}<0$ in two adjacent loops, we know that the total power has just crossed $\bar{P}T$. Thus, the step size is reduced so that it can get closer to the optimal solution. Besides, the step size is reduced more quickly (divided by 7 as in the algorithm) than when it is increased (doubled). Furthermore, the factor of the decrease and the increase are chosen to be relatively prime numbers. This insures that the algorithm will converge more quickly, while avoiding endless loops.

\section{Numerical Results}\label{sec:4}
In this section, we provide some numerical results to illustrate the proposed power allocation scheme. The system bandwidth is $W=5$ MHz and average transmitting power is $\bar{P}=5$ dBW. Suppose that the minimum distance $d_0=100$ m, the cellular radius is $R_0=2.5$ km and train velocity is $v=300$ km/h, then we have $T=100$ s. Let the pathloss exponent be $\alpha=4$. For Algorithm \ref{Alg:1}, the initial step size  is  $\Delta\lambda=0.01$ and the maximum allowable $r_{\Delta P}$ is $\varepsilon_{\Delta P}=0.001$.
\begin{figure}[!t]
\centering
\includegraphics[width=3.2in]{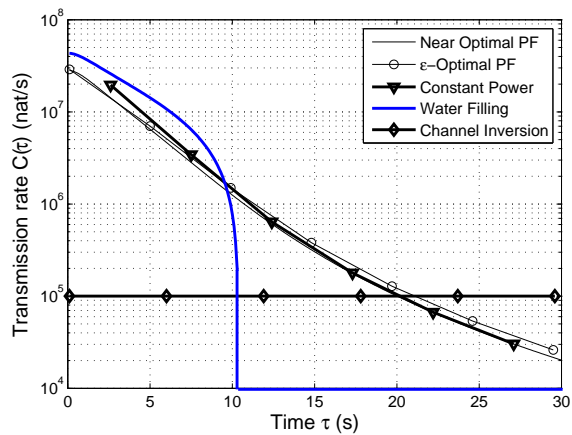}
\caption{Transmission rate ($C_\tau$) under different power allocations, $P=5$ dBW, $R_0=2.5$ km ($T=30$s)}
\label{fig:Ctau30}
\end{figure}

\begin{figure}[!t]
\centering
\includegraphics[width=3.2in]{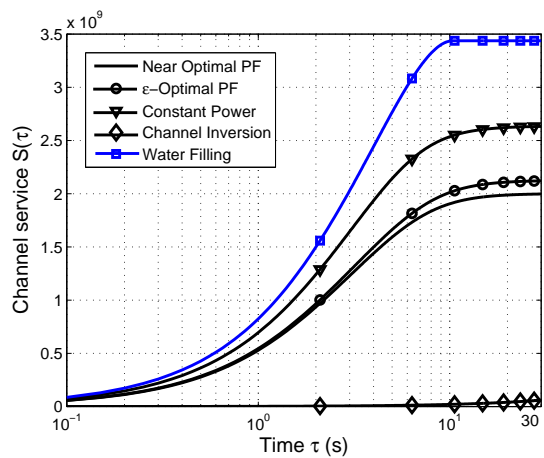}
\caption{The channel service versus time $\tau$, $P=5$ dBW, $R_0=2.5$ km ($T=30$s)}
\label{fig:Stau30}
\end{figure}

The power allocation (PA) functions along time via water-filling, channel inversion, as well as our proposed proportional fair (PF) scheme are presented. in Fig. \ref{fig:Ctau30}. Firstly, it is seen that the channel inversion PA (Eqn. \ref{eq:Ctau_inv}) get the best fairness along time by providing a stable data rate. On the other hand, the water-filling PA (Eqn. \ref{eq:Ctau_wf}) destroys the fairness completely by transmit nothing when $\tau>10.4$s. However, for $\tau<10$s , the transmission rate of the water-filling PA is much larger than the channel inverse PA. Nevertheless, a good  trade-off between efficiency and fairness is achieved by the proposed PF PA along the time, where a basic transmission is guaranteed anywhere within the cell. Particularly, the $\epsilon$-optimal PF PA performs better than the near optimal PF PA, especially at the cell edge regions. It is also seen that the constant PA (Eqn. \ref{eq:Ctau_cons}) is not so different from the proposed PF PA. Actually, the difference between the two will be further reduced when the cell radius gets larger. In this sense, the constant power PA is also a good PA scheme in terms of efficiency and fairness.

This is also seen in Fig. \ref{fig:Stau30}, where the channel service process is presented (the x-axis is in log-scale). In this figure, the comparison between the PA schemes is more clear, where the transmission rate can be checked by the gradient of corresponding curves. Thus, the PF PA is good for its satisfying capability and smooth change of transmission rate.
\begin{figure}[!t]
\centering
\includegraphics[width=3.2in]{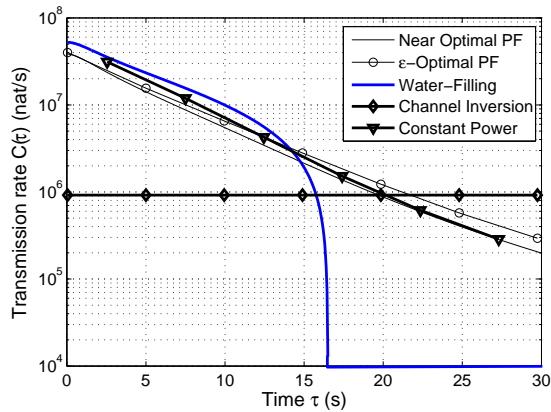}
\caption{Transmission rate ($C_\tau$) under different power allocations, $P=15$ dBW, $R_0=2.5$ km ($T=30$s)}
\label{fig:Ctau30_15}
\end{figure}
\begin{figure}[!t]
\centering
\includegraphics[width=3.2in]{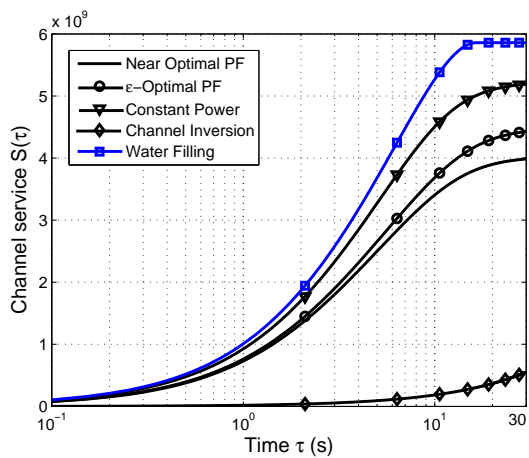}
\caption{The channel service versus time $\tau$, $P=15$ dBW, $R_0=2.5$ km ($T=30$s)}
\label{fig:Stau30_15}
\end{figure}

Fig. \ref{fig:Ctau30_15} and Fig. \ref{fig:Stau30_15} reconsidered the instantaneous capacity and channel service process as functions of time, where transmit power is increased to 15 dBW. Overall, previous conclusions still holds. In addition, it is seen that the performance of channel inversion and water-filling PA get improved. Specifically, the channel inversion PA achieves higher data rate and is not so bad as lower transmit power case (when $P=5$ dBW as in Fig. \ref{fig:Ctau30}) compared with other PA schemes. The water-filling also performs better by having a larger breakdown time. Finally, the proposed PF PA scheme also achieves better fairness compared with the constant PA. In short, higher transmit power provides more freedom to PA schemes. Therefore, Every PA scheme performs better if transmit power is increased, except the constant power PA, since there exists no power allocation.

Finally, the the convergence of Algorithm \ref{Alg:1} is given in Fig. \ref{fig:iter}. As shown, the algorithm is convergent and efficient.
\begin{figure}[!t]
\centering
\includegraphics[width=3.2in]{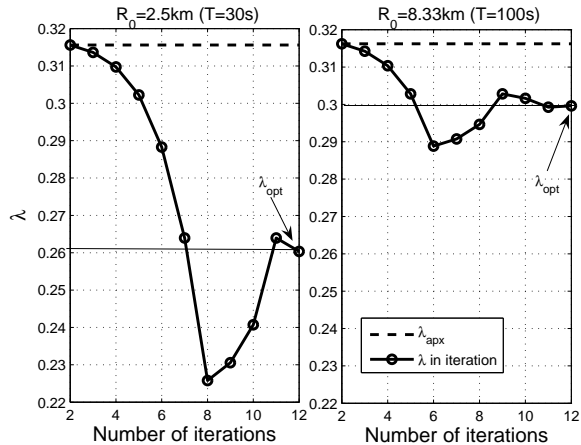}
\caption{The convergence of $\lambda$ in the iteration by Algorithm 1} \label{fig:iter}
\end{figure}

\section{Conclusion}\label{sec:5}
Different from traditional mobile communication (especially the cellular communication), very high mobility and very large communication distance become the main challenges in the high speed railway communications. These features make it indispensable to implement power allocation along the time. Based on the concept of the channel service, we investigated in this paper the trade-off between channel service and fairness along time, by the proposed proportional fairness along time for HSR systems. Both near optimal PF allocation scheme and algorithm finding the $\epsilon$-optimal PF allocation are given, which were shown to be efficient via numerical results. It is also seen that the constant power allocation is a good scheme in terms of efficiency and fairness, especially for large cell and low transmit power cases.

\appendix
\textit{Proof of Lemma \ref{lem:2_pfhsr}:}
The proof follows a similar way of the proof of the equivalence between the PF criteria (\ref{df:pf_org}) and the maximization problem (\ref{prob:1}) for the multi user case \cite{Kim-2005}. Under some power allocation scheme $P(\tau)$ and $S(\tau)$, denote $x_\tau=C^P(\tau)$ and $y_\tau=C^S(\tau)$ for short, respectively. Then lemma \ref{lem:2_pfhsr} will be proved if the solution to $\underset{P(\tau)}{\max} \int_0^T \ln x_\tau \mathrm{d}\tau$ satisfies $\int_0^T\frac{y_\tau - x_\tau}{x_\tau} \mathrm{d}\tau\leq0$ for any $y_\tau$, and conversely, a proportional fairness $x_\tau$ must be the solution to the optimization problem.

It is seen that the functional $\int_0^T \ln x_\tau \mathrm{d}\tau$ is strictly concave since the functional $\ln x_\tau$ is strictly concave. Therefore, the functional has its unique maximization point.

Suppose $x_\tau$ is the maximization point of the functional $\int_0^T \ln x_\tau \mathrm{d}\tau$ and $y_\tau=x_\tau+\Delta x_\tau$ deviates very slightly from $x_\tau$. Then we have
\begin{equation}\nonumber
    \begin{split}
        0\geq&\int_0^T \ln y_\tau \mathrm{d}\tau-\int_0^T \ln x_\tau \mathrm{d}\tau\\
        =&\int_0^T \left(\ln x_\tau\right)'(y_\tau-x_\tau) \mathrm{d}\tau\\
        =&\int_0^T \frac{y_\tau-x_\tau}{x_\tau} \mathrm{d}\tau,
    \end{split}
\end{equation}
i.e., $x_\tau$ is proportional fair.

Conversely, suppose $x_\tau$ is proportional fair, $\int_0^T \left(\ln x_\tau\right)'$ $(y_\tau-x_\tau) \mathrm{d}\tau=\int_0^T \frac{y_\tau-x_\tau}{x_\tau} \mathrm{d}\tau\leq0$ will hold for any other $y_\tau$. Since functional $\int_0^T \ln x_\tau \mathrm{d}\tau$ is strictly concave, we have $\int_0^T \left(\ln x_\tau\right)''\Delta x_\tau^2 \mathrm{d}\tau\leq0$. Recall that the functional $\ln y_\tau$ is also derivable for all orders and can be expanded in Taylor's series as $\ln y_\tau=\ln x_\tau + \sum_{n=1}^\infty \frac{(\ln x_\tau)^{(n)}}{n!} (y_\tau-x_\tau)^n$.  For some $\Delta x_\tau\rightarrow0$, we have
\begin{equation}\nonumber
    \begin{split}
        &\int_0^T \ln y_\tau \mathrm{d}\tau\\
        =&\int_0^T \left[ \ln x_\tau+\left(\ln x_\tau\right)'\Delta x_\tau+\frac12\left(\ln x_\tau\right)''\Delta x_\tau^2 + o(x_\tau^2) \right] \mathrm{d}\tau\\
        \leq&\int_0^T \left[ \ln x_\tau+ \frac{y_\tau-x_\tau}{x_\tau} \right] \mathrm{d}\tau
    \end{split}
\end{equation}
which leads to
\begin{equation}\nonumber
    \int_0^T \left[\ln y_\tau - \ln x_\tau\right] \mathrm{d}\tau\leq \int_0^T \left[  \frac{y_\tau-x_\tau}{x_\tau} \right] \mathrm{d}\tau \leq0.
\end{equation}
Therefore, $x_\tau$ is the unique maximization point. This completes the proof of lemma \ref{lem:2_pfhsr}.

\section*{Acknowledgement}
Prof. P. Fan's work was partly supported by the China Major
State Basic Research Development Program (973 Program) No.
2012CB316100(2), National Natural Science Foundation of China
(NSFC) No. 61171064  and NSFC No. 61021001.
Prof. K. B. Letaief's work was partly supported by RGC under grant No. 610311.

\bibliographystyle{IEEEtran}

\begin{thebibliography}{11}
\bibitem{Zhao-2013}
W. Luo,  X. Fang, M.Cheng and Y. Zhao, ``Efficient Multiple Group Multiple Antenna (MGMA) Scheme for High Speed Railway Viaduct," \textit{IEEE Trans. Veh. Tech.}, early access, 2013.

\bibitem{Qiu-2012}
L. Yang, G. Ren, B. Yang and Z. Qiu, ``Fast Time-Varying Channel Estimation Technique for LTE Uplink in HST Environment", \textit{IEEE Trans. Veh. Tech.}, vol. 61, no. 9, pp: 4009-4019, 2012.

\bibitem{Colombaioni-2008}
S. Leva, A.P. Morando and P. Colombaioni, ``Dynamic Analysis of a High-Speed Train",
 \textit{IEEE Trans. Veh. Tech.}, vol. 57, no. 1, pp 107-119, Jan. 2008.

\bibitem{Zhou-2011}
Y. Zhou, Z. Pan, J. Hu, J. and X. Mo, ``Broadband wireless communications on high speed trains,'' \textit{Wireless and Optical Commun. Conf.} (WOCC), 2011.


\bibitem{Kelly-1998}
F. P. Kelly, A.K. Maulloo and D.K.H. Tan, ``Rate control in communication
networks: shadow prices, proportional fairness and stability", \textit{Journal
of the Operational Research Society}, vol. 49, no. 3, pp. 237-252, Mar., 1998.

\bibitem{Kim-2005}
H. Kim and Y. Han, ``A proportional fair scheduling for multicarrier
transmission systems," \textit{IEEE Commun. Letters}, vol. 9, no. 3, pp. 210-
212, Mar. 2005.


\bibitem{DWF-Deterministic-TWC}
Y. Dong, Q. Wang, P. Fan and K. B. Letaief, ``The Deterministic Time-Linearity of Service Provided by Fading Channels," \textit{IEEE Trans. Wireless Commun.}, vol. 11, no. 5, pp. 1666-1675, May, 2012.

\bibitem{Chuang-bs_arrange}
C. Zhang, P. Fan, Y. Dong and K. Xiong, ``Channel Service Based High-Speed Railway Base Station Arrangement," The 1st \textit{International Workshop on High Mobility Wireless Communications}, Chengdu, China, 2012.

\bibitem{YF-DFO}
Y. Yang, P. Fan, ''Doppler frequency offset estimation and diversity reception
scheme of high-speed railway with multiple antennas on separated carriage'', \textit{Journal of Modern Transportation}, vol. 20, no. 4, pp. 227-233, Dec. 2012.

\bibitem{Lambert_w}
"Lambert W function," online available at
http://en. wikipedia.org /wiki /Lambert$_-$W$_-$function

\end{thebibliography}

\end{document}